\documentclass[12pt]{article}
\usepackage{amsfonts}
\usepackage{graphicx}
\usepackage{amsmath}
\usepackage{amsfonts}
\usepackage{amssymb}
\usepackage{amsthm}

\newtheorem{theorem}{Theorem}[section]
\newtheorem{lemma}[theorem]{Lemma}
\newtheorem{proposition}[theorem]{Proposition}
\theoremstyle{definition}
\newtheorem{criterion}[theorem]{Criterion}

\newtheorem{example}[theorem]{Example}
\newtheorem{remark}[theorem]{Remark}
\renewenvironment{proof}[1][Proof]{\noindent\textbf{#1.} }{\hfill \rule{0.5em}{0.5em}}

\begin{document}

\title{Non-traded call's volatility smiles}
\author{Marek Capinski}
\date{\today}
\maketitle

\begin{abstract}
Real life hedging in the Black-Scholes model must be imperfect and if the stock's drift is higher than the risk free rate, leads to a profit on average. Hence the option price is examined as a fair game agreement between the parties, based on expected payoffs and a simple measure of risk. The resulting prices result in the volatility smile. 
\end{abstract}

\section{Introduction}

We consider a European call option written on stock satisfying the
Black-Scholes equation%
\begin{equation}
dS(t)=\mu S(t)dt+\sigma S(t)dW_{P}(t),  \label{eq:BS-equ}
\end{equation}%
where $W_{P}(t)$ is a Wiener process in the probability space $(\Omega ,%
\mathcal{F},P)$. The strike price is $K$ and exercise time $T$, and the
option payoff is denoted $C(T)=(S(T)-K)^{+}.$

The option is sold over the counter for the price $C$. This option is
assumed to be non-tradable, so the arbitrage pricing argument does not
apply and the price will the result of an agreement between the writer
(seller) and the holder (buyer). The buyer is an investor who strongly
believes that the stock will perform well and is ready to invest in a call,
which gives certain leverage. In particular, he believes that $\mu >r,$
where $r$ is the risk-free rate, and we make this assumption throughout. 

The writer will hedge by taking a position in the primary market consisting
of stock $S$ and the money market account $A(t)=e^{rt}.$ The model is
complete, but this requires continuous rebalancing, which is impossible. So
the initial position will be taken for a period of time, suppose first this
is kept constant till exercise, with $x$ shares. As a result, the expected
writer's payoff, computed with respect to the physical probability, which
means the stock price with $\mu $ is used, is larger than the expected
option payoff (since the stock on average grows faster than the risk free
asset). The holder may require the initial option price be lower, so that
the expected payoffs are the same. However, this equilibrium price depends
on $x,$ so one more condition is needed. A natural assumption is that
writer's hedging risk is minimised, which requires choosing a risk measure.
The resulting $x$ is not binding and only serves the purpose of negotiating
the option price.

At any time the writer may change the hedging position and the same analysis
(based on the assumption that the new position hypothetically stays static)
will give the current value of the option.

It is interesting to see that the implied volatility (that is, the $\sigma $
which would result in the Black-Scholes option price being the same as the
equilibrium price) shows the volatility smile effect.

A similar strategy is discussed in \cite{CJP}, where imperfect hedging leads
to the prices found by means of the physical probability but the method is
different (it involves utility functions). Various risk measures are used
for imperfect hedging in \cite{ST} and \cite{WZS}. A classical example of a
non-traded option are employee's options, but the key feature is concerned
with the exercise time limitations, so our approach does not cover this. 

\section{Pricing by the Black-Scholes formula}

Despite the fact that the price is not based on the no-arbitrage principle,
it is a natural idea to use the Black-Scholes formula first.

Let $r$ be the risk-free rate and introduce the risk-neutral probability $Q,$
where $W_{Q}(t)=W(t)+\frac{\mu -r}{\sigma }t$ is a Wiener process under $Q.$
The stock prices follow the equation%
\begin{equation*}
dS(t)=rS(t)dt+\sigma S(t)dW_{Q}(t)
\end{equation*}%
and the no-arbitrage call option price is given by the Black-Scholes formula 
\begin{equation*}
C_{\mathrm{BS}}=S(0)N(d_{+}^{r})-e^{-rT}KN(d_{-}^{r}),
\end{equation*}%
where%
\begin{equation*}
d_{\pm }^{r}=\frac{\ln \frac{S(0)}{K}+rT\pm \frac{1}{2}\sigma ^{2}T}{\sigma 
\sqrt{T}}.
\end{equation*}

Suppose that $C=C_{\mathrm{BS}}.$ At time $T$ the option holder's profit is 
\begin{equation*}
P_{\mathrm{H}}(T)=C(T)-C_{\mathrm{BS}}e^{rT}.
\end{equation*}%
The expectation of $P_{H}(T)$ with respect to risk-neutral probability is of
course zero since $C_{\mathrm{BS}}=e^{-rT}\mathbb{E}_{Q}(C(T))$. However,
the risk-neutral world is abstract and in reality the stock follows equation
(\ref{eq:BS-equ}) under the physical probability $P$, so the expectation
with respect to this measure is relevant:%
\begin{equation*}
\mathbb{E}_{P}(P_{\mathrm{H}}(T))=\mathbb{E}_{P}(C(T))-C_{\mathrm{BS}}e^{rT}.
\end{equation*}

\begin{proposition}
If $\mu >r,$ then $\mathbb{E}_{P}(P_{\mathrm{H}}(T))>0.$
\end{proposition}

\begin{proof}
It is sufficient to see that $\mathbb{E}_{Q}(C(T))<\mathbb{E}_{P}(C(T)).$ We
have 
\begin{eqnarray*}
\mathbb{E}_{Q}(C(T)) &=&\mathbb{E}_{Q}(\exp (rT-\frac{1}{2}\sigma
^{2}T+\sigma W_{Q}(T)), \\
\mathbb{E}_{P}(C(T)) &=&\mathbb{E}_{P}(\exp (\mu T-\frac{1}{2}\sigma
^{2}T+\sigma W_{P}(T)),
\end{eqnarray*}%
but the distribution of $W_{Q}(T)$ under $Q$ is the same as the distribution
of $W_{P}(T)$ under $P,$ so%
\begin{equation*}
\mathbb{E}_{Q}(\exp (\sigma W_{Q}(T))=\mathbb{E}_{P}(\exp (\sigma W(T))
\end{equation*}%
which gives the claim.
\end{proof}

Now we look at this from the option writer's perspective. This amount $C_{%
\mathrm{BS}}$ obtained for the option is used to build a hedging portfolio $%
(x,y)$ with $x$ being the number of shares and $y$ the number of units of
the money market account. The risk free position is then $y=C_{\mathrm{BS}%
}-xS(0).$

Idealy, the hedging position should be continuously rebalanced but this is
not realistic. A practical hedging strategy could be piece-wise constant
with random rebalancing times. However, since the option price must be
agreed now, it is natural to consider first a constant strategy.

Assuming that the positions $x,y$ are kept constant over the time interval $%
[0,T],$ the terminal value of the portfolio is%
\begin{equation*}
V(T)=x[S(T)-S(0)e^{rT}]+C_{\mathrm{BS}}e^{rT}
\end{equation*}%
so the profit of the option writer $P_{\mathrm{W}}^{C}(T)=V(T)-C(T)$ is 
\begin{equation*}
P_{\mathrm{W}}(T)=x[S(T)-S(0)e^{rT}]+C_{\mathrm{BS}}e^{rT}-C(T).
\end{equation*}

For any $x$ the expected writer's profit computed with respect to the
risk-neutral probability is zero. Indeed, we have $\mathbb{E}%
_{Q}(S(T))=S(0)e^{rT}$ and $C_{\mathrm{BS}}=\mathbb{E}%
_{Q}(e^{-rT}(S(T)-K)^{+}),$ which gives the claim. But we use the physical
probability and since $\mathbb{E}_{P}(S(T))=S(0)e^{\mu T},$ the expected
option writer's profit is an increasing function of $x,$ provided $\mu >r,$
in fact it is a linear function with positive slope: 
\begin{equation*}
\mathbb{E}_{P}(P_{\mathrm{W}}(T))=xS(0)(e^{\mu T}-e^{rT})+C_{\mathrm{BS}%
}e^{rT}-\mathbb{E}_{P}(C(T)).
\end{equation*}

\begin{example}
\label{Ex;first}Assume $S(0)=100,$ $\mu =10\%,$ $\sigma =20\%,$ $r=5\%$ and
consider European call with strike $K=100,$ $T=1.$ The Black-Scholes price
is $C_{\mathrm{BS}}=10.45$ and using the formula $\mathbb{E}%
_{P}(C(T))=e^{\mu T}N(d_{+}^{\mu })-KN(d_{-}^{\mu })$ where $d_{\pm }^{\mu }$
is defines as $d_{\pm }^{r}$ with $\mu $ replacing $r,$ we find $\mathbb{E}%
_{P}(P_{\mathrm{H}}(T))=3.58.$ With delta hedging $x=N(d_{+}^{r})$ we have $%
\mathbb{E}_{P}(P_{\mathrm{W}}(T))=-0.25.$ However, it is writer's decision
to choose $x$ and if he tries to maximise the expected profit he may take
large $x,$ but this would be an active and risky position typical for an
investor, not an option writer, whose priority is hedging.
\end{example}

\section{Equilibrium pricing}

As we have seen, the price given by the Black-Scholes formula may not be
acceptable for the option writer. The option price of an OTC transaction has
to be determined by an agreement between both parties involved and we
propose some natural criteria. Recall that the expected profits of the
parties involved, the holder and the writer respectively, are%
\begin{eqnarray*}
\mathbb{E}_{P}(P_{\mathrm{H}}(T)) &=&\mathbb{E}_{P}(C(T))-Ce^{rT}, \\
\mathbb{E}_{P}(P_{\mathrm{W}}(T)) &=&xS(0)(e^{\mu T}-e^{rT})+Ce^{rT}-\mathbb{%
E}_{P}(C(T)).
\end{eqnarray*}

\begin{criterion}
(\textbf{Fair play) }The price $C$ should be such that $\mathbb{E}_{P}(P_{%
\mathrm{W}}(T))=\mathbb{E}_{P}(P_{\mathrm{H}}(T)).$
\end{criterion}

This condition gives us the price as a function of $x$: 
\begin{equation}
C_{x}=e^{-rT}(\mathbb{E}_{P}(C(T))-\frac{1}{2}xS(0)(e^{\mu T}-e^{rT})).
\label{eq:C(x)}
\end{equation}%
We need a mutually acceptable criterion for establishing $x,$ and it is
natural to assume that it will be based on minimising writer's risk. There
are many choices of a risk measure, like Value-at-Risk, or
Conditional-Value-at-Risk (CVaR) (see \cite{PTRM} for instance), but to use
them we have to agree on some confidence level, which is subjective. So we
propose to use a simple and natural idea of expected loss under the
condition that it is positive.

The writer's loss is given by 
\begin{equation}
L_{\mathrm{W}}(T)=-x(S(T)-S(0)e^{rT})-Ce^{rT}+C(T)  \label{eq:loss-writer}
\end{equation}%
and we introduce the conditional expectation of the loss, provided it is
positive:%
\begin{equation*}
\gamma _{\mathrm{W}}(x)=\mathbb{E}_{P}(L_{\mathrm{W}}(T)|L_{\mathrm{W}%
}(T)>0)=\frac{\mathbb{E}_{P}(L_{\mathrm{W}}(T)\mathbf{1}_{\{L_{\mathrm{W}%
}(T)>0\}})}{P(L_{\mathrm{W}}(T)>0)}
\end{equation*}%
and seek for the minimum. In fact, $\gamma _{\mathrm{W}}(x)$ is $\mathrm{CVaR%
}_{\alpha }(L_{\mathrm{W}}(T))$ at $\alpha =P(L_{\mathrm{W}}(T)>0).$

\begin{criterion}
(\textbf{Risk minimising) }The hedging position $x_{0}$ should be such that $%
\gamma _{W}(x)$ considered in the interval $[0,1]$ attains minimum at $x_{0}$
$.$
\end{criterion}

Note that the option writer decides the value of $x,$ and the number
resulting from this criterion is not binding.

In order to investigate the existence of the prices and hedging portfolios
satisfying the above criteria, we derive a closed form expression for the
function $\gamma _{W}(x).$

\begin{proposition}
\label{Prop:formofgamma}Assume $x\in \lbrack 0,1).$ Writer's risk is given by

\begin{enumerate}
\item 
\begin{equation*}
\gamma _{\mathrm{W}}(x)=\frac{\mathbb{E}(C(T)\mathbf{1}_{\{L_{\mathrm{W}%
}>0\}})}{P(L_{\mathrm{W}}(T)>0)}-x\frac{\mathbb{E}(S(T)\mathbf{1}_{\{L_{%
\mathrm{W}}>0\}})}{P(L_{\mathrm{W}}(T)>0)}+xS(0)e^{rT}-C_{x}e^{rT},
\end{equation*}%
where

\item $P(L_{\mathrm{W}}(T)>0)=N(d_{1})+1-N(d_{2}),$ with%
\begin{eqnarray*}
d_{1} &=&\frac{1}{\sigma \sqrt{T}}\left( \ln \frac{(xS(0)-C_{x})e^{rT}}{S(0)x%
}-\mu T+\frac{1}{2}\sigma ^{2}T\right) , \\
d_{2} &=&\frac{1}{\sigma \sqrt{T}}\left( \ln \frac{K+(C_{x}-xS(0))e^{rT}}{%
S(0)(1-x)}-\mu T+\frac{1}{2}\sigma ^{2}T\right) ,
\end{eqnarray*}

\item $\mathbb{E}(C(T)\mathbf{1}_{\{L_{\mathrm{W}}(T)>0\}})=S(0)e^{\mu
T}[1-N(d_{2}-\sigma \sqrt{T})]-K[1-N(d_{2})],$

\item $\mathbb{E}(S(T)\mathbf{1}_{\{L_{\mathrm{W}}(T)>0\}})=S(0)e^{\mu
T}[N(d_{1}-\sigma \sqrt{T})+1-N(d_{2}-\sigma \sqrt{T})].$
\end{enumerate}
\end{proposition}

The following relation will be used in the proof.

\begin{lemma}
\label{lemma}We have $d_{1}<d<d_{2},$ where 
\begin{equation*}
d=\frac{1}{\sigma \sqrt{T}}\left( \ln \frac{K}{S(0)}-\mu T+\frac{1}{2}\sigma
^{2}T\right) .
\end{equation*}
\end{lemma}

The proof is routine and a sketch is given in the appendix. It is based on
the inequality $\mu >r$ and employs call-put parity.

\begin{proof}[Proof of Proposition \protect\ref{Prop:formofgamma}]
1. To find the form of $\gamma _{\mathrm{W}}(x)$ all we have to if to use
the expression (\ref{eq:loss-writer}) for $L_{\mathrm{W}}(T)$ with $C=C_{x}$%
, insert it to the definition of conditional probability, and use the
linearity of expectation.

2. Write $S(T)=S(0)e^{\mu T-\frac{1}{2}\sigma ^{2}T+\sigma \sqrt{T}Z}$ where 
$Z\sim N(0,1).$

First note that $L_{\mathrm{W}}(T)>0$ is equivalent to $%
x(S(T)-S(0)e^{rT})+C_{x}e^{rT}-C(T)<0$ and we consider two cases relevant
for the form of $C(T).$

Case 1. $S(T)<K,$ that is $Z<d,$ and here $C(T)=0$ so that 
\begin{equation*}
L_{\mathrm{W}}(T)=x(S(T)-S(0)e^{rT})+C_{x}e^{rT}.
\end{equation*}%
Next, $L_{\mathrm{W}}(T)>0$ is now equivalent to $S(T)<\frac{1}{x}%
(xS(0)-C_{x})e^{rT},$ which in turn corresponds to $Z<d_{1}$. As a result we
get the probability $N(d_{1}).$

Case 2. $S(T)>K$ that is $Z>d,$ and here $C(T)=S(T)-K.$ Now, simple algebra
shows that $L_{\mathrm{W}}(T)>0$ if and only if $Z>d_{2}$ with probability $%
1-N(d_{2}),$ and these two cases give%
\begin{equation*}
L_{\mathrm{W}}(T)>0\quad \text{iff\quad }Z\leq d_{1}\text{ or }Z>d_{2},
\end{equation*}%
disjoint events, hence the claim.

3. For $\mathbb{E}(C(T)\mathbf{1}_{\{L_{\mathrm{W}}(T)>0\}})$ note that $%
C(T)=(S(0)e^{\mu T-\frac{1}{2}\sigma ^{2}T+\sigma \sqrt{T}Z}-K)\mathbf{1}%
_{\{Z>d\}}$ while $\mathbf{1}_{\{L_{\mathrm{W}}(T)>0\}}=\mathbf{1}%
_{\{Z>d_{2}\}\cup \{Z\leq d_{1}\}}$ so we have (using the lemma)%
\begin{equation*}
\mathbb{E}(C(T)\mathbf{1}_{\{L_{\mathrm{W}}(T)>0\}})=\mathbb{E}((S(0)e^{\mu
T-\frac{1}{2}\sigma ^{2}T+\sigma \sqrt{T}Z}-K)\mathbf{1}_{\{Z>d_{2}\}}).
\end{equation*}%
Routine integration gives 
\begin{equation}
\frac{1}{\sqrt{2\pi }}\int_{d_{2}}^{\infty }e^{\sigma \sqrt{T}z-\frac{1}{2}%
\sigma ^{2}T}e^{-\frac{1}{2}z^{2}}dz=[1-N(d_{2}-\sigma \sqrt{T})]
\label{eq:integral}
\end{equation}%
which readily imples the claim.

4. For stock prices we have $\mathbb{E}(S(T)\mathbf{1}_{\{L_{\mathrm{W}%
}(T)>0\}})=\mathbb{E}(S(T)\mathbf{1}_{\{Z>d_{2}\}\cup \{Z\leq d_{1}\}})$ and
all that is left to compute two integrals similar to (\ref{eq:integral}).
\end{proof}

Option holder may wish to find the risk resulting from the option writer's
decision and the form of the risk function will be needed. Recall that
holder's loss is $L_{\mathrm{H}}(T)=C_{x}e^{rT}-C(T)$ and the risk is
similarly assumed to be the conditional expectation $\gamma _{\mathrm{H}}(x)=%
\mathbb{E}(L_{\mathrm{H}}(T)|L_{\mathrm{H}}(T)>0).$

\begin{proposition}
We have 
\begin{equation*}
\gamma _{\mathrm{H}}(x)=C_{x}e^{rT}-\frac{S(0)e^{\mu T}[N(d^{\prime }-\sigma 
\sqrt{T})-N(d-\sigma \sqrt{T})]-K[N(d^{\prime })-N(d)]}{N(d^{\prime })},
\end{equation*}%
where%
\begin{equation*}
d^{\prime }=\frac{1}{\sigma \sqrt{T}}\left( \ln \frac{K+C_{x}e^{rT}}{S(0)}%
-\mu T+\frac{1}{2}\sigma ^{2}T\right) 
\end{equation*}
\end{proposition}

\begin{proof}
Holder's loss is positive if $S(T)<C_{x}e^{rT}+K$ which corresponds to $%
Z<d^{\prime }$ (note that $d<d^{\prime })$ and gives the probability of the
condition. Next 
\begin{eqnarray*}
\mathbb{E}(L_{\mathrm{H}}(T)\mathbf{1}_{\{Z<d^{\prime }\}})
&=&C_{x}e^{rT}P(Z<d^{\prime })-\mathbb{E}((S(T)-K)\mathbf{1}_{\{Z<d^{\prime
}\}}\mathbf{1}_{\{Z>d\}}) \\
&=&C_{x}e^{rT}(Z<d^{\prime })-\mathbb{E}(S(T)\mathbf{1}_{\{d<Z<d^{\prime
}\}})+K[N(d^{\prime })-N(d)]
\end{eqnarray*}%
and we can compute the expectation in a similar way as before.
\end{proof}

With explicit form of $\gamma _{\mathrm{W}}(x)$ we could try to prove the
existence of the minimum in $[0,1].$ However, the form of the derivative is
quite complicated and solving the necessary condition would have to be
performed numerically, so it is best to restrict our attention to an example
where the minimum is found numerically.

\begin{example}
With the data as in Example \ref{Ex;first} we find the form of the risk
function and sketch the graph. The witer's risk is smallest for $x=0.7212$
and the corresponding option price is $C_{x}=12.10,$ higher than the
Black-Scholes price.%

\begin{figure}[h]
\centering{\includegraphics[width=4in]{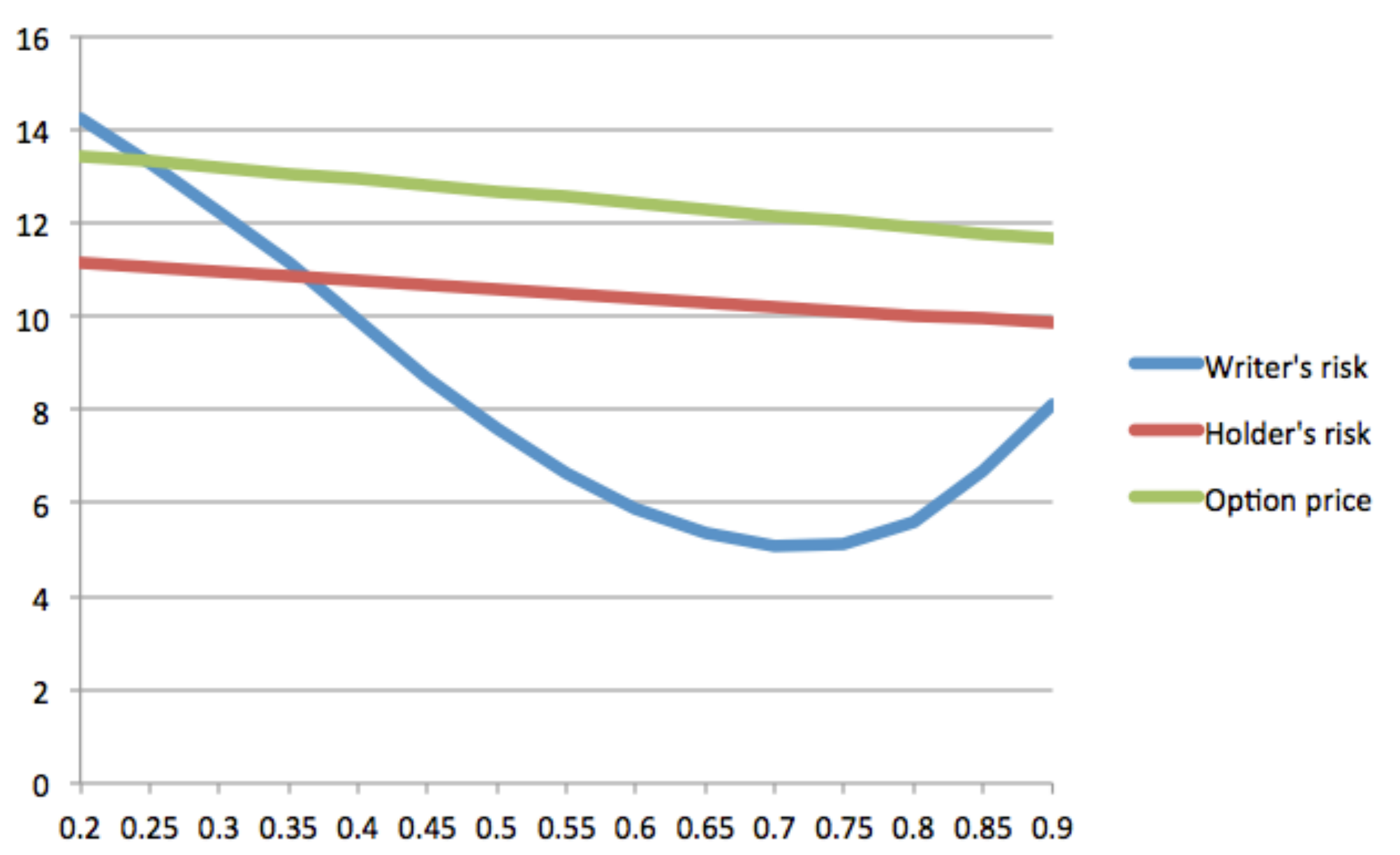}}\label
{risk.pdf}
{\caption{Risk functions}}
\end{figure}
\end{example}

Finally, we repeat this for various strike prices and compute the implied
volatility using the Black-Scholes formula.

\begin{example}
With the same data as before we get the following results (see Figure 2).%
\begin{equation*}
\begin{array}{ccccccc}
K & 90 & 95 & 100 & 105 & 110 & 115 \\ 
C & 18.89 & 15.28 & 12.10 & 9.38 & 7.12 & 5.30 \\ 
\sigma & 27.43\% & 25.67\% & 24.38\% & 23.42\% & 22.72\% & 22.20\%%
\end{array}%
\end{equation*}

\begin{figure}[h]
\centering{\includegraphics[width=3in]{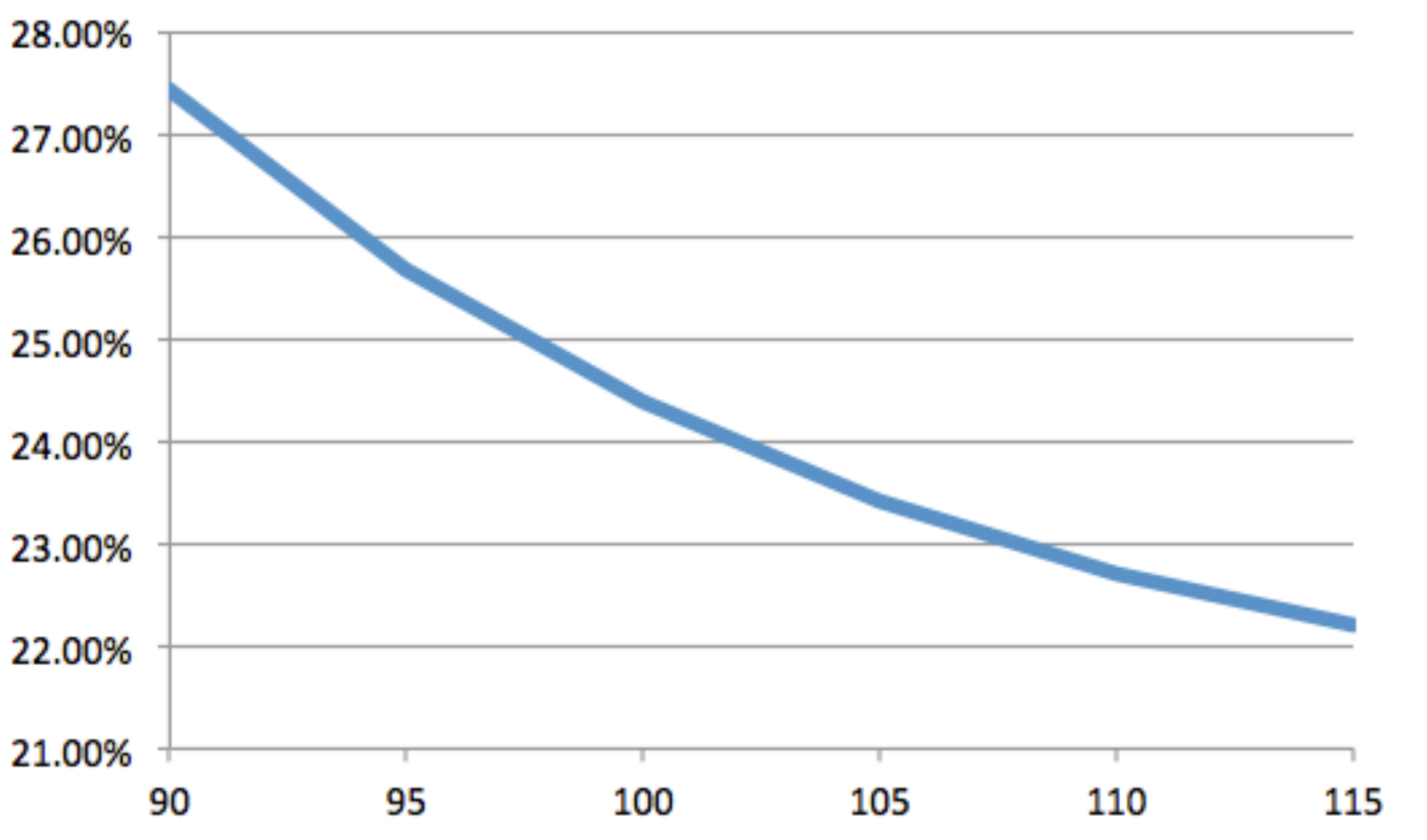}}\label
{smile.pdf}
{\caption{Volatility smile}}
\end{figure}%
\end{example}

\begin{remark}
The option writer may wish to rebalance the strategy in the future and if
the same risk measure is accepted, all that is needed is to use the
expressions derived above for the relevant time period, replacing $T$ by $%
T-t $, where $t$ is the time of rebalancing.
\end{remark}

\section{Appendix}

\begin{proof}[Proof of Lemma \protect\ref{lemma}]
We have to see that%
\begin{equation*}
\frac{(xS(0)-C_{x})e^{rT}}{S(0)x}<\frac{K}{S(0)}<\frac{K+(C_{x}-xS(0))e^{rT}%
}{S(0)(1-x)}.
\end{equation*}%
First we demostrate that the first inequality holds for $x=1.$ Using the
definition of $C_{x}$ we have to show that%
\begin{equation*}
S(0)e^{rT}-\mathbb{E}_{P}(C(T))+\frac{1}{2}S(0)(e^{\mu T}-e^{rT})<K.
\end{equation*}%
We use the call-put payoffs parity $C(T)=P(T)+S(T)-K,$ and insert this above
on the left to get $K-\mathbb{E}_{P}(P(T))-\frac{1}{2}S(0)(e^{\mu T}-e^{rT}),
$ clearly smaller than $K.$ The claim follows from the fact that $\frac{%
(xS(0)-C_{x})e^{rT}}{S(0)x}$ is increasing with $x$ which follows after
inserting the definition of $C_{x}$ and perfoming some algebra: 
\begin{equation*}
\frac{(xS(0)-C_{x})e^{rT}}{S(0)x}=e^{rT}-\frac{\mathbb{E}(C(T))}{S(0)x}+%
\frac{1}{2}(e^{\mu T}-e^{rT}).
\end{equation*}

For the second inequality note that it is obviously true for $x=0$ with $%
C_{0}=e^{-rT}(\mathbb{E}(C(T))>0.$ It is sufficient to see that the function 
$\frac{K+(C_{x}-xS(0))e^{rT}}{S(0)(1-x)}$ is an incresing function of $x.$
To get this we compute its derivative and some algebra shows it is positive.
\end{proof}

\end{document}